\documentclass[letterpaper, 10 pt, conference]{ieeeconf}  

\IEEEoverridecommandlockouts                              

\usepackage{graphicx}      

\usepackage{multicol}
\usepackage{amsmath,amssymb}
\usepackage{amsfonts}
\usepackage{textcomp}
\usepackage{psfrag}
\usepackage{multimedia}
\usepackage{fancybox}

\newcommand{\matr}[2]{\left[\begin{array}{@{}#1@{}}#2\end{array}\right]}
\newcommand{\E}[2]{\mathbb{E}_{#1}\left[#2\right]}

\newtheorem{Theorem}{Theorem}
\newtheorem{Lemma}[Theorem]{Lemma}
\newtheorem{Corollary}[Theorem]{Corollary}

\newtheorem{Remark}{Remark}

\usepackage[usenames,dvipsnames]{color}
\definecolor{wheat}{rgb}{0.96,0.87,0.70}
\definecolor{mario}{rgb}{0.8,0.8,1}
\definecolor{seb}{rgb}{0.8,1,0.8}

\title{\LARGE \bf
	Practical Reinforcement Learning of Stabilizing Economic MPC
}

\author{Mario Zanon$^{1}$, S\'ebastien Gros$^{2}$ and Alberto Bemporad$^{1}$
	\thanks{$^{1}$Mario Zanon and Alberto Bemporad are with the IMT School for Advanced Studies Lucca, Italy. 
		{\tt\small \{name.surname\}@imtlucca.it}}%
	\thanks{$^{2}$S\'ebastien Gros is with the Department of Electrical Engineering, Chalmers University of Technology, H\"{o}rsalsv\"{a}gen 9a, Sweden.
		{\tt\small grosse@chalmers.se}}%
}

\begin{document}

\maketitle
\thispagestyle{empty}
\pagestyle{empty}

\begin{abstract}                
	Reinforcement Learning (RL) has demonstrated a huge potential in learning optimal policies without any prior knowledge of the process to be controlled. 
	Model Predictive Control (MPC) is a popular control technique which is able to deal with nonlinear dynamics and state and input constraints. The main drawback of MPC is the need of identifying an accurate model, which in many cases cannot be easily obtained. Because of model inaccuracy, MPC can fail at delivering satisfactory closed-loop performance. Using RL to tune the MPC formulation or, conversely, using MPC as a function approximator in RL allows one to combine the advantages of the two techniques.
	This approach has important advantages, but it requires an adaptation of the existing algorithms. We therefore propose an improved RL algorithm for MPC and test it in simulations on a rather challenging example.
\end{abstract}


\section{Introduction}
Reinforcement learning (RL) is a \emph{model-free} control technique which recursively updates the controller parameters in order to achieve optimality. 
Once the controller parameters have been learned, the controller will implement the control action which yields an infinite-horizon optimal cost  for each initial condition~\cite{Sutton2018}. RL has drawn increasing attention thanks to the striking results obtained in beating chess and go masters~\cite{Silver2016}, and in learning how to make a robot walk or fly without supervision~\cite{Abbeel2007,Wang2012a}.

In order to be able to solve the RL problem in practice, function approximation strategies must be employed. Such approximations typically rely on a set of basis functions, or \emph{features}, and corresponding parameters that multiply them. In this case, the function approximator is linear in the parameters. However, nonlinear function approximators are commonly deployed using Deep Neural Networks (DNN).

Model Predictive Control (MPC) is a \emph{model-based} technique which exploits a model of the system dynamics to predict the system's future behavior and optimize a given performance index, possibly subject to input and state constraints~\cite{Rawlings2009b,Grune2011,Borrelli2017}. The success of MPC is due to its ability to enforce constraints and yield optimal trajectories. While a plethora of efficient algorithms for the online solution of MPC problems has been developed, the main drawback of this control technique is the need of identifying the open-loop model offline, which is typically the most time-consuming phase of control design.

The advantages of MPC and RL can be combined together by framing MPC as a function approximator within an RL context. The tuning parameters of the MPC problem (e.g., cost weighting matrices, model parameters, etc.) can be framed as parameters of a nonlinear function approximator, namely the MPC optimization problem. This setup can be adopted to approximate the feedback policy, the value function, or the action-value function.

The main advantages of combining MPC with RL are (a) the ease of introducing a constraint-enforcing policy within RL, (b) the possibility of improving existing models and controllers by using them as an initial guess in RL, and (c) the possibility of interpreting the learned algorithm in a model-based framework.

The combination of learning and control techniques has been proposed in, e.g.,~\cite{Koller2018,Aswani2013,Ostafew2016,Berkenkamp2017}. Some attempts at combining RL and the linear quadratic regulator have been presented in~\cite{Lewis2009,Lewis2012}. To the best of our knowledge, however, \cite{Gros2018} is the first work proposing to use NMPC as a function approximator in RL.

In this paper, building on the ideas from~\cite{Gros2018}, we further analyse the combination of RL and MPC. The main contribution of this paper is the development of an improved algorithmic framework tailored to the problem formulation which overcomes some shortcomings of basic RL algorithms. We demonstrate the potential of our approach in simulations on an involved nonlinear example from the process industry.


This paper is structured as follows. Section~\ref{sec:q_learning} introduces $Q$-learning. The main contributions of this paper are presented in Section~\ref{sec:mpc_rl}, which discusses the use of MPC as a function approximator in RL, and Section~\ref{sec:data_efficient}, which presents an algorithm adaptation, tailored to the problem. Some numerical examples are given in Section~\ref{sec:examples}. The paper is concluded by Section~\ref{sec:conclusions} which also outlines future research directions.

\section{Reinforcement Learning}
\label{sec:q_learning}


Consider a Markov Decision Process (MDP) with state transition dynamics $\mathbb{P}[s_+|s,a]$, where $s$ and $a$ denote the states and actions (or controls) respectively. We also introduce the stage cost $\ell(s,a)$ and the discount factor $0< \gamma \leq 1$. 
The  action-value function $Q_\star(s,a)$ and value function $V_\star(s)$ associated with the optimal policy $\pi_\star\left(s\right)$ are defined by the Bellman equations:
\begin{subequations}
	\label{eq:Bellman}
	\begin{align}
	Q_\star\left(s,a\right) &= \ell\left(s,a\right) + \gamma \mathbb{E}\left[V_\star(s_+)\,|\, s,a\right], \label{eq:Bellman1}\\
	V_\star\left(s\right) &= Q_\star\left(s,\pi_\star\left(s\right)\right) = \min_{a}\, Q_\star\left(s,a\right). \label{eq:Bellman2}
	\end{align}
\end{subequations}
$Q$-learning parametrizes the action value function as $Q_\theta(s,a)$, where $\theta$ is a vector of parameters whose values have to be learned, and aims at minimizing $\| Q_\star(s,a) - Q_\theta (s,a) \|_2^2$.
Given a state-action pair $(s,a)$ and the next state $s_+$, standard algorithms~\cite{Sutton2018} update the parameter using
\begin{subequations}
	\label{eq:Qlearning}
	\begin{align}
		\delta &= \ell(s,a) + \gamma\,\min_{a_{+}} Q_{\tilde \theta}(s_{+},a_{+}) - Q_\theta(s,a), \label{eq:TDError}\\
		\theta &\leftarrow \theta + \alpha\delta\nabla_\theta Q_\theta(s,a), \label{eq:Qlearning:Update}
	\end{align}
\end{subequations}
where $\delta$ is known as the \emph{Temporal-Difference (TD) error}. In batch policy updates $\tilde \theta$ is kept constant for $N_\mathrm{upd}$ steps after which it is updated as $\tilde \theta\leftarrow\theta$. In instantaneous policy updates $N_\mathrm{upd}=1$ such that $\tilde \theta$ is updated at every time instant.

It is important to underline that, while in this paper we focus on $Q$-learning (which has been successfully deployed on some applications, e.g.,~\cite{Watkins1989,Mnih2015,Theocharous2015}), our approach can be readily deployed on other TD-algorithms such as SARSA~\cite{Kober2013}. Even in approaches directly optimizing the policy, learning the action-value function is often necessary~\cite{Sutton1999,Silver2014}, such that the developments of this paper apply to a fairly broad class of RL algorithms.

While a very common choice is to use Deep Neural Networks (DNN) as function approximators, it is hard to analyze closed-loop stability in DNN-based RL. Moreover, in case a controller is already available, information on how to control the system cannot be easily incorporated in the problem formulation. For this reason, we propose to use MPC as a function approximator, since it makes it possible to enforce closed-loop stability guarantees and an existing controller can be used as initial guess for the algorithm. 



\section{MPC-Based RL}
\label{sec:mpc_rl}

The use of MPC to parametrize the action-value function has been first advocated in~\cite{Gros2018}. 
In the following, we first present the function approximation used and recall the most important result in Theorem~\ref{thm:wrong_model}, then we further analyze the properties of MPC-based RL and in the next section we will propose a new variant of the  $Q$-learning algorithm tailored to MPC-based RL.

\subsection{Parametrization of the Function Approximations}

We parametrize the action-value function using an MPC scheme of the form
\begin{subequations}
	\label{eq:param_nmpc}
	\begin{align}
	Q^N_\theta(s,a) =  \min_{z}\ & \lambda_\theta(x_0)+\gamma^N V^\mathrm{f}_\theta(x_N) + \sum_{k=0}^{N-1} \gamma^k  \ell_\theta(x_k,u_k) \nonumber \\ 
	& + \sum_{k=0}^{N}\gamma^k \left ( s_k^\top W_s s_k +  w_s^\top s_k\right ) \label{eq:param_nmpc:cost}\\
	\mathrm{s.t.} \ & x_0 = s, \quad u_0=a, \label{eq:param_nmpc:initial}\\
	&x_{k+1} = f_\theta\left(x_k,u_k\right), \label{eq:param_nmpc:dynamics}\\
	& g_\theta\left(u_k\right) \leq 0, \label{eq:param_nmpc:input_const} \\
	& h_\theta\left(x_k,u_k\right) \leq s_k,\quad h^\mathrm{f}_\theta(x_N) \leq s_N, \label{eq:param_nmpc:path} 
	\end{align}
\end{subequations}
where $z=(x_0,u_0,s_0,\ldots,x_N)$. Consequently, we obtain
\begin{align*}
\pi_\theta^N(s) = \mathrm{arg}\min_a\, Q_\theta^N(s,a), \qquad V_\theta^N(s) = \min_a\, Q_\theta^N(s,a).
\end{align*}
Note that the policy $\pi_\theta^N(s)$ and value function $V_\theta^N(s)$ are equivalently obtained by solving Problem~\eqref{eq:param_nmpc} with constraint $u_0=a$ removed. 
In order to address feasibility issues in Problem~\eqref{eq:param_nmpc}, in the presence of state-dependent constraints we adopt an exact relaxation of such constraints~\cite{Scokaert1999a}. 

\begin{Remark}
	MPC formulations in which the stage cost is not positive-definite are commonly referred to as \emph{Economic MPC}.
	In order to enforce closed-loop stability and the existence of the MPC solution, we assume that $\ell_\theta$ and $V^\mathrm{f}_\theta$ are positive definite functions. Therefore, in order to deal with the situation in which the true stage cost $\ell$ is not positive-definite, we have introduced the arrival penalty $\lambda_\theta$.
	For all details on this topic, we refer to~\cite{Gros2018} and references therein.
%
\end{Remark}
\begin{Remark}
	The proposed setup readily accommodates for formulations in which the cost penalizes deviations from a given reference. In that case, the stage cost (both in RL and MPC) depends on a reference, passed to the problem as an exogenous signal. Input-output model formulations also readily fit in the proposed framework.
\end{Remark}

Among the desirable properties of the proposed formulation, we mention nominal stability guarantees, and the possibility to introduce constraints accounting for, e.g., actuator limitations, safe operation of the system, etc. While being aware that the current framework is not able to fully exploit these advantages, with the presented developments we aim at constructing a sound basis which will be used in future research with the intent of developing self-tuning, safe, and stable economic MPC controllers.

\subsection{Learning the Model: RL and System Identification}

We recall the following fundamental theorem from~\cite{Gros2018}, which states that the optimal value and action-value functions as well as the optimal policy can be learned even by using MPC based on a state transition model $\hat f$ which is different from the true model $f$. This also entails that there is no guarantee (and no need) that the RL algorithm will learn a physically meaningful model.
\begin{Theorem}[\cite{Gros2018}]
	\label{thm:wrong_model}
	Consider a given (possibly stochastic) state transition model $\hat s_{k+1}=\hat f(\hat s_k,\hat a_k)$, possibly different from the true model $f(s_k,a_k)$. 
	Define the optimal value function 
	\begin{align*}
		\hat V^N(s) =\min_{\pi}\E{}{\gamma^N V_\star(\hat s^{ \hat \pi^N}_N) + \sum_{k=0}^{N-1}\, \gamma^k \hat \ell( \hat s^{ \hat \pi^N}_k, \hat \pi( \hat s^{ \hat \pi^N}_k))} 
	\end{align*}	
	associated with stage cost $\hat \ell$, and the terminal cost $V_\star(s)$ over an optimization horizon $N$. Define $\hat s^{\hat \pi}_{0,\ldots,N}$ as the (possibly stochastic) trajectories of the state transition model $\hat f$ under a policy $\hat \pi$, starting from $\hat s^{\hat \pi}_0 = s$, and $\hat \pi^N$ the optimal policy associated with $\hat V_N(s)$ and $\hat Q^N(s,a)$ the associated action-value function. Consider the set $\mathcal{S}$ such that
	\begin{align*}
	\left|\,\mathbb{E}\left[V_\star\left(\hat s^{\pi_\star}_{k}\right)\right]\,\right| < \infty, \qquad \forall \, s\in\mathcal{S}. 
	\end{align*}
	Then, $ \exists \ \hat \ell$ such that the following identities hold on $\mathcal{S}$:
	\begin{itemize}
		\item[(i)] 
		$\hat V^N(s) = V_\star(s)$ 
		\item[(ii)] $
		\hat \pi^N \left(s\right) = \pi_\star\left(s\right)$
		\item[(iii)] $
		\hat Q_N\left(s,a\right) = Q_\star\left(s,a\right)$ for the inputs $a$ such that $\left |\,\mathbb{E}\left[V^\star\left(\hat s_+\right)\,|\, s,a\right] \,\right | < \infty$.
	\end{itemize} 
\end{Theorem}

We can now further clarify the theorem and formalize a form of ``orthogonality'' between reinforcement learning and system identification.
\begin{Corollary}
	\label{cor:rl_sysId}
	Let us split the parameter as $\theta=(\theta_f,\theta_Q)$, such that the model $f_\theta$ only depends on $\theta_f$ and
	assume a perfect parametrization, such that, under adequate exploration of the state-action space, RL learns $Q$ and $V$ perfectly with $\hat \theta \neq \theta$, i.e.,
	\begin{align*}
	\delta_k=\ell(s_k,u_k) + \gamma \E{}{V^N_{\hat\theta}(s_{k+1})} - Q^N_{\hat\theta}(s_k,a_k)=0.
	\end{align*}
	Then, leaving the task of identifying $\theta_f$ to a separate identification algorithm is not detrimental to the learning of the optimal value, action-value function, and policy.
\end{Corollary}
%

Unfortunately, Corollary~\ref{cor:rl_sysId} applies only to the case of perfect parametrization. In practice, function approximation with imperfect parametrization and lack of excitation can destroy this form of orthogonality. Ongoing research is currently further investigating such aspects.

\subsection{On the Parametrization of $V$ and $Q$}
\label{sec:VQparametrisation}

Since the Bellman principle of optimality states that $Q_\theta^N(s,a)=\ell_\theta(s,a) + V^{N-1}_\theta(f_\theta(s,a))$, one could be tempted to 
replace $V_\theta^{N}(s)=\min_a Q_\theta^N(s,a)$ by $V_\theta^{N-1}(s)$
in the computation of the TD error. We discuss in the following lemma why this approach can minimize the TD error while not learning the correct action-value function and, consequently, not delivering the optimal policy. 

\begin{Lemma}
	Consider computing the TD error~\eqref{eq:TDError} by
	replacing $\min_a Q_\theta^N(s,a)$ by $V_\theta^{N-1}(s)$ to obtain
	\begin{align}
		\delta = \ell(s,a) + \gamma\,V_\theta^{N-1}(f(s,a)) - Q_\theta^N(s,a).
		\label{eq:wrong_td_error}
	\end{align}
	Then, it is possible to obtain $\delta=0$ without minimising the error $\|Q^\star(s,a)-Q_\theta^N(s,a)\|^2$ and, therefore, without learning the optimal policy.
\end{Lemma}
\begin{proof}
	We prove the Lemma by a simple counter example.
	Consider the LQR case, with $\theta=(\hat A, \hat B,\hat P)$ and
	\begin{align*}
	Q_\theta^1(s,a) &= \matr{c}{s \\ a}^\top\hspace{-2pt} \matr{ll}{T + \gamma \hat A^\top \hat P \hat A & S  + \gamma \hat A^\top \hat P \hat B \\ S^\top  + \gamma \hat B^\top \hat P \hat A \ \ & R + \gamma \hat B^\top \hat P \hat B} \matr{c}{s \\ a}, \\ 
	V_\theta^0(s) &= s^\top  \hat Ps. 
	\end{align*}
	Then,~\eqref{eq:wrong_td_error} reads as
	\begin{align*}
	\delta=\gamma \matr{c}{s \\ a}^\top \hspace{-4pt} \left (\matr{cc}{A & \hspace{-4pt}B}^\top \hspace{-3pt} \hat P\matr{cc}{A & \hspace{-4pt}B} - \matr{cc}{\hat A & \hspace{-4pt}\hat B}^\top\hspace{-3pt} \hat P \matr{cc}{\hat A & \hspace{-4pt}\hat B}\right ) \matr{c}{s \\ a}.
	\end{align*}
	Assume now that the correct model is identified, i.e., $\hat A=A$, $\hat B=B$. Then, the TD error is independent of $\hat P$, which can be chosen as desired. This implies that, even though we use a perfect parametrization of the action-value function, $Q_\theta^1(s,a) \neq Q_\star(s,a)$. Consequently, the learned feedback given by 
	\begin{align}
	\label{eq:lqr_feedback_rl}
	\hat K=(R+\gamma \hat B^\top \hat P \hat B )^{-1}(S^\top + \hat B^\top \hat P \hat A)
	\end{align}
	is not uniquely defined. 
\end{proof}

For comparison, we consider now the standard case, in which the TD error is computed using
\begin{align*}
V_\theta^1(s) &= \min_a Q_\theta^1(s,a) \\
&=s^\top (T + \gamma \hat A^\top \hat P \hat A - (S + \gamma\hat A^\top \hat P \hat B)\hat K)s, 
\end{align*}
where $\hat K$ is given by~\eqref{eq:lqr_feedback_rl}. 
In this case, if $\hat A=A$, $\hat B=B$, $\hat P$ must solve the algebraic Riccati equation associated with the given stage cost and model. It is important to stress, however, that infinitely many solutions exist in general such that the true model will not be identified. As an example we provide $\theta_1=(0.5\hat A, 0.5\hat B,4\hat P)$: it is immediate to verify that $Q_{\theta_1}^1(s,a) \equiv Q_\theta^1(s,a)$.

\subsection{Condensed Model-Free Parametrization}
Motivated by the previous considerations on the possibility of learning a wrong model, we propose next a formulation which is truly model-free as it does not include model parameters as coefficients to be learned. For simplicity we only focus here on the case of linear dynamics with a quadratic cost. We define $z = (x_0,u_0,\ldots,u_{N-1})$; then
\begin{subequations}
	\label{eq:param_nmpc:valuefunction2}
	\begin{align}
	Q_\theta^N(s,a) =  \min_{z}\  & z^\top \hspace{-4pt} M z  + m^\top z + c \\
	\mathrm{s.t.} \  & x_0 = s, \quad u_0=a, \\
	& C z \leq d, 
	\end{align}
\end{subequations}
where $\theta=(M,C,d,m)$ and $M$ is a symmetric (typically dense) matrix. Problem formulation~\eqref{eq:param_nmpc:valuefunction2} is sometimes used in MPC and optimal control, in which case the parameters $\theta$ are related to the system dynamics and cost~\cite{Bock1984}.
In this case, the dependence of $V$ and $Q$ on the parameters is less nonlinear than in the model-based parametrization, see~\eqref{eq:param_nmpc}. However, the introduction of a prediction horizon $N>1$ possibly introduces more parameters than the standard formulation. A special case occurs when there are only input constraints and they are known: in this case $C$ can be fixed and only $M$, $m$ need to be learned. 
The LQR case simplifies to
\begin{align}
Q_\theta^1(s,a) =  \matr{c}{s\\ a}^\top \matr{ll}{\bar T & \bar S\\ \bar S^\top & \bar R} \matr{c}{s\\ a},
\end{align}
such that $K=\bar R^{-1}\bar S^\top$ and $\pi^\star(s)=-Ks$


\section{Algorithm}
\label{sec:data_efficient}

The main practical difficulties related to using MPC as a function approximator in RL are that: (a) the function approximator is nonlinear in the parameters $\theta$, and (b) the MPC problem is guaranteed to have a meaningful solution only if the cost is positive-definite. In this section we propose an algorithm that addresses both issues.

We start by recalling that the main motivation for the stochastic gradient approach typically used in $Q$-learning stems from an equivalence with the table-lookup case, i.e., when the state-action space is discrete and the action-value function is parametrized as
\begin{align*}
	Q_\theta(s,a) = \sum_{\bar s\in \mathcal{S}, \bar a\in \mathcal{A}} \theta_{\bar s,\bar a} \phi_{\bar s,\bar a}(s,a),
\end{align*}
and $\phi_{\bar s,\bar a}(s,a)=1$ if $\bar s=s, \bar a=a$ and $0$ otherwise. In this case, $Q_\theta$ is linear in the parameter $\theta$. Parameter $\alpha$ in~\eqref{eq:Qlearning:Update} is used in order to approximate the expected value of the TD error: one can roughly interpret it as the inverse of the amount of samples over which the average is computed. 

In this case, the update~\eqref{eq:Qlearning:Update} obtained with the (exact) function approximation matches that of the enumeration
\begin{align*}
	\sum_{\bar s\in \mathcal{S}, \bar a\in \mathcal{A}} \theta_{\bar s,\bar a} \delta_{\bar s,\bar a}(s,a) = \alpha \delta.
\end{align*}
The update~\eqref{eq:Qlearning:Update} can also be written as $\theta \leftarrow \theta + \alpha \Delta \theta^*$ with
\begin{align*}
	\Delta \theta^* &= \arg\min_{\Delta \theta} \ (\delta_k - \nabla_\theta Q_\theta \Delta \theta)^2.
\end{align*}

We remark that this is not the case if $\phi$ is not normalized or if $Q$ is a nonlinear function of the parameter $\theta$. This also implies that parameter $\alpha$ loses its original meaning, as it is also used to dampen the stochastic gradient step in order to enforce convergence.

We propose to apply the update $\theta \leftarrow \theta + \alpha (\theta^*-\theta)$ where $\theta^*$ is the solution of
\begin{subequations}
	\label{eq:q_solution}
	\begin{align}
	\min_{\theta}  \ & \sum_{k=0}^n \left (\ell(s_k,a_k) + \gamma\,\min_{a^\prime} Q_{\tilde \theta}^N(s_{k+1},a^\prime)- Q_{\theta}^N(s_k,a_k) \right )^2
	\label{eq:q_solution_cost}\\
	\mathrm{s.t.} \ & \nabla^2 \ell_{ \theta} \succ 0, \qquad \nabla^2 V^\mathrm{f}_{ \theta} \succ 0. \label{eq:q_pd_constr}
	\end{align}
\end{subequations}
Note that we solve Problem~\eqref{eq:q_solution} to full convergence at each step. In case of a linear parametrization of $Q_\theta$, convergence is obtained in one step by using a Gauss-Newton Hessian approximation, in case Constraints~\eqref{eq:q_pd_constr} are inactive. This also directly solves the issue of scaling, present in the table-lookup case if $\phi$ is not normalized. Finally, if instead of solving Problem~\eqref{eq:q_solution} to full convergence constraints~\eqref{eq:q_pd_constr} are neglected and one takes a single full Newton step, the update~\eqref{eq:Qlearning:Update} is recovered.


We highlight next the main features of Formulation~\eqref{eq:q_solution}:
\begin{itemize}
	\item \textit{Globalization}: globalization strategies such as line search ensure descent and, therefore, convergence. Descent is typically not enforced in stochastic gradient approaches such as~\eqref{eq:Qlearning}, which could take a step which yields a larger TD error than with the previous parameter estimate.
	\item \textit{Positive-definiteness enforcement}: Constraints~\eqref{eq:q_pd_constr} guarantee that the cost is positive-definite and, therefore, let the MPC problem be stabilizing and well-posed. Without this constraint, even with a well-posed, positive-definite cost, throughout the RL iterates one might obtain an indefinite cost yielding an unbounded solution. 
	We have written the positive-definiteness Constraints~\eqref{eq:q_pd_constr} using the Hessian of $\ell$ and $V^\mathrm{f}$, which is reasonable if both functions are quadratic. However, other approaches relying on richer function approximations are possible, e.g. using sum-of-squares techniques. 
	\item \textit{Best fit}: by solving Problem~\eqref{eq:q_solution} to full convergence, the step $\theta^*-\theta$ is the one which minimizes $\delta$, analogously to the lookup-table case. Therefore, the choice of parameter $\alpha$ can be done purely by considerations about the expected value approximation.
\end{itemize}
In summary, the main differences with the standard update~\eqref{eq:Qlearning:Update} are: (a) positive-definiteness enforcement, (b) insensitivity to parameter scaling and (c) guarantee of improvement at each step through globalization and full convergence.

\subsection{Derivative Computation}
We detail next how to compute the derivatives of the action-value function with respect to the parameters. 
To this end, we define the Lagrangian function underlying NMPC problem~\eqref{eq:param_nmpc} as
\begin{subequations}
	\begin{align*}
	\mathcal{L}_\theta^N(s,y) = \ & \gamma^N V^\mathrm{f}_\theta(x_N)  + \chi_0^\top\left(x_0 - s\right) + \mu_N^\top h^\mathrm{f}_\theta(x_N)\\ 
	& + \sum_{k=0}^{N-1}  \chi_{k+1}^\top\left(f_\theta\left(x_k,u_k\right) - x_{k+1} \right) + \nu_k^\top g_\theta\left(u_k\right)\\
	&+\gamma^k \ell_\theta(x_k,u_k) + \mu_{k}^\top h_\theta\left(x_k,u_k\right) + \zeta^\top(u_0-a),
	\end{align*}
\end{subequations}
where $\chi,\mu,\nu,\zeta$ are the multipliers associated with constraints \eqref{eq:param_nmpc:initial}-\eqref{eq:param_nmpc:path} and $y=(z,\chi,\mu,\nu,\zeta)$. Note that, for $\zeta=0$, $\mathcal{L}_\theta^N(s,y)$ is the Lagrangian function associated to the NMPC problem defining the value function $\min_a Q_\theta(s,a)$.
We observe that \cite{Buskens2001}
\begin{align}
\label{eq:NMPCQgradient}
\nabla_\theta Q_\theta^N(s,a) = \nabla_\theta \mathcal{L}_\theta^N(s,y^\star)
\end{align}
holds for $y^\star$ given by the primal-dual solution of \eqref{eq:param_nmpc}. Note that this equality holds because constraints \eqref{eq:param_nmpc:initial} are not an explicit function of $\theta$. The gradient \eqref{eq:NMPCQgradient} is therefore straightforward to build as a by-product of solving the NMPC problem \eqref{eq:param_nmpc}. We additionally observe that
\begin{align}
\label{eq:NMPCVgradient}
\nabla_\theta \min_a Q_\theta^N(s,a) = \nabla_\theta \mathcal{L}^N(s,y^\diamond),
\end{align}
where $y^\diamond$ is given by the primal-dual solution to \eqref{eq:param_nmpc} with constraint $u_0=a$ removed and $\zeta^\diamond=0$.

We remark that second-order derivatives can also be computed, though in general they depend on the derivative of the optimal primal-dual solution with respect to the parameters, 
given by
\begin{align}
\label{eq:dpolicy_dtheta}
\nabla_{\theta} y^\star = - \nabla_{\theta } \xi_\theta^N(s,y^\star)\nabla_{y} \xi_\theta^N(s,y^\star)^{-1},
\end{align}
where $\xi_\theta^N(s,y)$ gathers the primal-dual KKT conditions underlying the NMPC scheme \eqref{eq:param_nmpc}. For a complete discussion on parametric sensitivity analysis of NLPs we refer to~\cite{Buskens2001} and references therein.

%

%

\section{Numerical Example}
\label{sec:examples}

We consider an example from the process industry, i.e. the evaporation process modelled in~\cite{Wang1994,Sonntag2006} and used in~\cite{Amrit2013a,Zanon2016b} to demonstrate the potential of economic MPC in the nominal case. 
For the sake of brevity we omit the model equations and non-quadratic economic cost function, which include states $(X_2,P_2)$ (concentration and pressure) and controls $(P_{100},F_{200})$ (pressure and flow). All details can be found in~\cite{Amrit2013a,Zanon2016b}.
The model further depends on concentration $X_1$, flow $F_2$, and temperatures $T_1,T_{200}$, which are assumed to be constant in the control model. In reality, these quantities are stochastic with variance $\sigma_{X_1}=1$, $\sigma_{F_1}=2$, $\sigma_{T_1}=8$, $\sigma_{T_{200}}=5$,  and mean centred on the nominal value. Bounds $(25,40)\leq (X_2,P_2) \leq (100,80)$ on the states and $100 \leq (P_{100},F_{200}) \leq 400$ on the controls are present. In particular, the bound $X_2\geq 25$ is introduced in order to ensure sufficient quality of the product. All state bounds are relaxed as in~\eqref{eq:param_nmpc:path}.

We parametrize an NMPC controller as in~\eqref{eq:param_nmpc}, i.e., a nonlinear non-condensed MPC formulation, 
with $N=10$, $\lambda_\theta, V^\mathrm{f}_\theta, l_\theta$ quadratic functions defined by Hessian $H_\dagger$, gradient $h_\dagger$, and constant $c_\dagger$, $\dagger=\{\lambda, V^\mathrm{f}, l\}$. The model is parametrized as the nominal model with the addition of a constant, i.e., $f_\theta(x,u) = f(x,u) + c_f$. The control constraints are fixed and the state constraints are parametrized as simple bounds, i.e., $h_\theta(x,u)=(x-x_\mathrm{l},x_\mathrm{u}-x)$. The vector of parameter therefore reads as:
\begin{align*}
\theta = ( H_\lambda, h_\lambda, c_\lambda, H_{V^\mathrm{f}}, h_{V^\mathrm{f}}, c_{V^\mathrm{f}}, H_l, h_l, c_l, c_f, x_\mathrm{l},x_\mathrm{u}).
\end{align*}
Constants $w_s=1$, $W_s=I$ are fixed and assumed to reflect the known cost of violating the state constraints.

We use the batch policy update with $\alpha=10^{-2}$ and update the parameters with the learned ones every $N_\mathrm{upd}=500$ time steps. In order to induce enough exploration, we use an $\epsilon$-greedy policy which is greedy $90\, \%$ of the samples, while in the remaining $10 \, \%$ we apply the action
\begin{align*}
a = \mathrm{sat}(e,u_\mathrm{l},u_\mathrm{u}), && e \sim \mathcal{N}(0,\sqrt{10}),
\end{align*}
where $\mathrm{sat}(\cdot,u_\mathrm{l},u_\mathrm{u})$ saturates the input between its lower and upper bounds $u_\mathrm{l},u_\mathrm{u}$, respectively.

\begin{figure}
	\begin{center}
		\includegraphics[width=\linewidth,clip,trim=0 145 0 120]{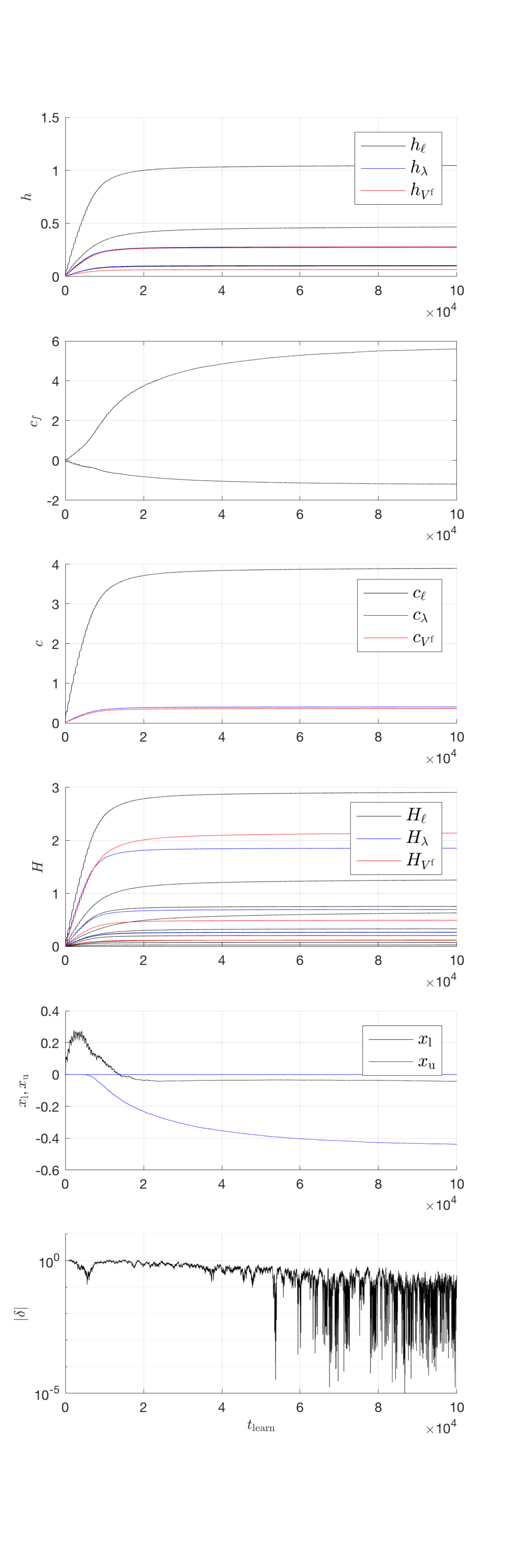}
	\end{center}
	\caption{Evolution of the parameters (increment w.r.t. the initial guess value) and of the TD error (averaged over the preceding $1000$ samples).}
	\label{fig:nmpc_learn}
	\vspace{-1em}
\end{figure}

\begin{figure}
	\begin{center}
		\includegraphics[width=0.9\linewidth,clip,trim=20 40 30 30]{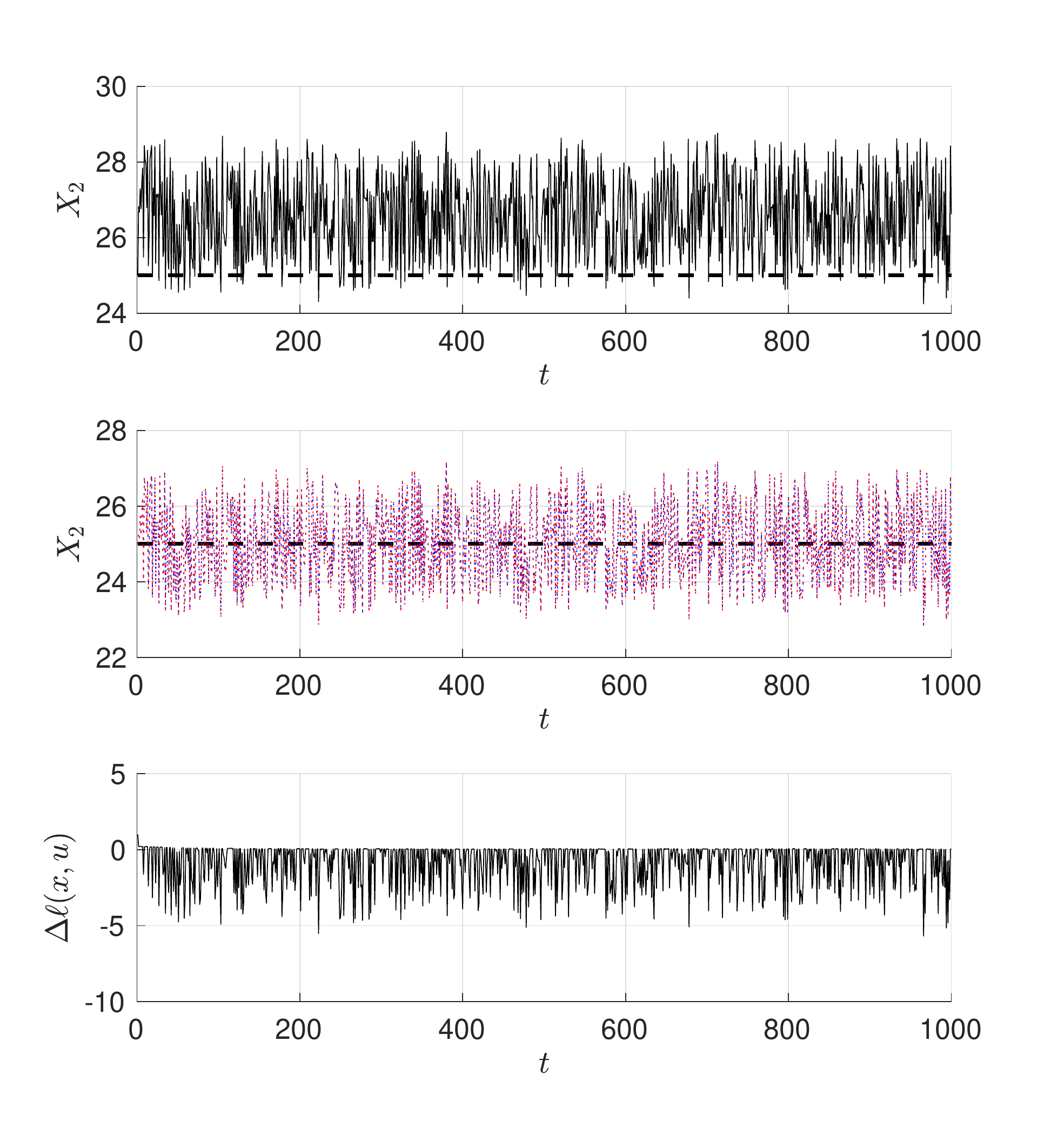}
	\end{center}
	\caption{NMPC closed-loop simulations. Top graph: concentration $X_2$ for NMPC with RL tuning. Middle graph: naive tuning (red) and nominal economic tuning (blue). In both graphs, the quality constraint is displayed in thick dashed black line. Bottom graph: difference of instantaneous cost between NMPC with RL tuning and naive tuning.}
	\label{fig:nmpc_x_l}
	\vspace{-1em}
\end{figure}

We initialize the ENMPC scheme by the 
naive initial guess $H_l=I$, $x_\mathrm{l}=(25,40)$, $x_\mathrm{u}=(100,80)$, while all other parameters are $0$. 
While at every step we do check that $H_l, H_{V^\mathrm{f}}\succ 0$, during the learning phase the parameters never violate this constraint. 
As displayed in Figure~\ref{fig:nmpc_learn}, the algorithm converges to a constant parameter value while reducing the average TD-error. If the standard parameter update~\eqref{eq:Qlearning:Update} is applied, RL does not converge with the proposed $\alpha$. If a smaller value is used, the algorithm does not diverge but the parameters $\theta$ are updated very slowly, making the approach impractical.

We performed a simulation to compare the RL-tuned NMPC scheme to the naive initial guess and the NMPC tuned using the economic-based approach proposed in~\cite{Zanon2016b}, which relies on the nominal model. The economic gain obtained by RL is approximately $14\,\%$ and $12\,\%$ respectively. The effectiveness of the economically-tuned NMPC had been demonstrated in~\cite{Zanon2016b} in the absence of stochastic perturbations. In the considered scenario we observe that, while the economically-tuned scheme still performs better than the naively-tuned one, the RL tuning is able to significantly outperform the two other NMPC schemes by explicitly accounting for the stochastic perturbations.

The concentration $X_2$ and the difference in instantaneous cost between the RL-tuned and the naively-tuned scheme are displayed in Figure~\ref{fig:nmpc_x_l}. In particular, one can see that RL is trying to stabilize the concentration $X_2$ to a value which is higher than the nominally optimal one: the optimum in the presence of perturbations is obtained as a compromise between the loss due to operating at $X_2>25$ and the cost of violating the constraint $X_2\geq 25$. Indeed, the constraint is violated but only rarely and by small amounts.

\section{Conclusions and Outlook}
\label{sec:conclusions}

In this paper we analyzed the use of reinforcement learning to tune MPC schemes, aiming at a self-tuning controller which guarantees stability, safety (i.e., constraint satisfaction), and optimality. In order to be able to apply this approach in practice, we have proposed an improved algorithm based on $Q$-learning and we have tested it in simulations.

Future work will consider several research directions including: (a) further improvements in the algorithmic framework with the aim of developing data-efficient algorithms; (b) developing new algorithms for other RL paradigms, such as policy gradient methods; (c) further investigating the combination of system identification and RL in order to best update the MPC parameters while guaranteeing safety.

\bibliographystyle{IEEEtran}
\bibliography{my_bib}

\end{document}